\documentclass[aps,
floatfix,superscriptaddress,a4paper]{revtex4}
\usepackage{graphicx,amsmath,bbm,mathrsfs,amssymb,psfrag}

\newtheorem{lemma}{Lemma}
\newtheorem{definition}{Definition}

\newtheorem{remark}{Remark}

\newcommand{\CI}{{\mathbb{C}}}

\newcommand{\RI}{{\mathbb{R}}}

\newcommand{\cH}{{\mathcal{H}}}

\newcommand{\cU}{{\mathcal{U}}}
\newcommand{\cW}{{\mathcal{W}}}

\evensidemargin \oddsidemargin
\begin{document}

\author{Fabio Benatti}
\affiliation{Department of Physics, University of Trieste $\&$ INFN, Sezione di Trieste,\\
Strada Costiera 11, I-34051-Trieste,Italy\\
Email: benatti@ts.infn.it}

\author{Laure Gouba}
\affiliation{
The Abdus Salam International Centre for
Theoretical Physics (ICTP),\\
 Strada Costiera 11,
I-34151 Trieste Italy \\
Email: lgouba@ictp.it}

\title{Classical limits of quantum mechanics on a non-commutative configuration space}

\begin{abstract}
We consider a model of non-commutative Quantum Mechanics given by
two harmonic oscillators over a non-commutative two dimensional
configuration space. We study possible ways of removing the non-commutativity
based on the classical limit context known as anti-Wick quantization.
We show that removal of non-commutativity from the configuration space and from the
canonical operators are not commuting operations.
\end{abstract}

\maketitle

\section{Introduction}
Since the beginning, the classical limit of quantum mechanics has
been of primary interest. The most suitable context where to study it is
provided by the notion of coherent states as in \cite{klaus}.

In this work, we study the classical limit not of a standard quantum system,
but of two quantum harmonic oscillators whose spatial coordinates are themselves
non-commuting operators with non-commutative parameter $\theta$ \cite{laure}.
For a more physical approach of the problem see \cite{euro}.
One immediately has various possibilities to go to the
limit of classical harmonic oscillators on the commutative configuration space $\RI^2$.
One can first go from the non-commutative configuration space to $\RI^2$ by letting $\theta\to 0$
and then remove  the quantumness by letting $\hbar\to 0$; one can remove quantumness first and get
to a non-quantum system over a non-commutative configuration space and then remove the residual
non-commutativity; finally, one can remove both non-commutativities together.

In order to study these possibilities, we use the quantization/de-quantization schemes known as
anti-Wick quantization \cite{monique}.
In such a scheme we first quantize a $C^*$ algebra of continuous functions with identity by means
of suitably constructed Weyl operators and corresponding Gaussian states that allow to set up a
positive unital map from functions to bounded operators and then de-quantize it by getting back
to functions via another positive unital map.
Combining them together one has a means to first let $\theta\to0$ and then $\hbar\to0$ and vice versa:
the main result is that the two procedures do not commute.

Further, we study a harmonic like dynamics of the two non-commutative quantum oscillators and show
that the asymmetry in the two limits is even stronger; letting $\theta\to0$ first regains the standard
quantum mechanics of two independent harmonic oscillators. However, letting $\hbar\to0$ first does not
leave any dynamics on the non-quantum system over the non-commutative configuration space.

We start by giving a brief review of the anti-Wick quantization in Section \ref{sec2}.
Then, in Section \ref{sec3}, we briefly recall the model of non-commutative quantum harmonic
oscillators in two dimensions; in \ref{sec4} we construct the Weyl operators and Gaussian states
on which the anti-Wick quantization is based and in Section \ref{sec5} we study the various classical limits.
The time evolution and its classical limits will be discussed in Section \ref{sec6}.

\section{Anti-Wick Quantization}
\label{sec2}

In this section we shall shortly review the classical limit of quantum mechanics in the algebraic setting known as
\emph{anti-Wick quantization}; this technique is based on the quasi-classical properties of
\emph{coherent states} whose definition and properties we shall also summarize.
For later extension to the non-commutative
quantum mechanical context,
we shall consider
the standard setting of a classical system with $s$ degrees of freedom described by a phase-space
$\RI^{2s}$ with canonical coordinates and momenta $r=(q,p)\in M$; $q$ and $p$ denote vectors in $\RI^s$
whose components satisfy the canonical Poisson-bracket relations $\{q_i,p_j\}=\delta_{ij}$.
From now, scalar products will be denoted by $(q,p)=\sum_{j=1}^sq_jp_j$ and by $\|r\|^2$ norm of vectors.

Let $\hat{r}=(\hat{q},\hat{p})$ be the $2s$-dimensional vector of quantized coordinates and momenta
operators acting on the Hilbert space $\mathcal{H}$ of square-summable functions over $\mathbbm{R}^s$.
They satisfy the Heisenberg commutation relations $[\hat{r}_i,\hat{r}_j]=i\hbar\Omega_{ij}$,
where $\Omega$ is the $2s\times 2s$ symplectic matrix
\begin{equation}
\label{sympl}
\Omega=\begin{pmatrix}0&1_{s\times s}\\
-1_{s\times s}&0\end{pmatrix}\ ,\quad 1_{s\times s} \quad\hbox{the $s\times s$ identity matrix}\ .
\end{equation}
A useful $C^*$ algebraic description of the quantized system in terms of bounded operators
on $\mathcal{H}$ makes use is of the unitary Weyl operators
\begin{equation}
\label{Weyl}
\hat{W}_\hbar(r)=\exp{\Big(\frac{i}{\hbar}(r, \Omega\hat{r})\Big)}=\exp{\Big(\frac{i}{\hbar}((q,\hat p)-(p,\hat q))\Big)}\ ,
\end{equation}
where $(\cdot\,,\cdot)$ denotes the scalar product over generic $2s$-dimensional vectors.
They satisfy the Weyl algebraic relations
\begin{equation}
\label{Weyl1}
\hat{W}_\hbar(r_1)\hat{W}_\hbar(r_2)=\exp{\Big(\frac{i}{2\hbar}(r_2,\Omega r_1)\Big)}\,\hat{W}_\hbar(r_1+r_2)\ ,
\end{equation}
whence they linearly span an algebra whose norm closure known as
\emph{Weyl algebra}.
One can pass from a real formulation whereby the Weyl operators are labelled by $r\in\RI^{2s}$ to a complex formulation where they are labelled by a complex vector $z\in\CI^s$: this is done by introducing creation and annihilation operators
\begin{equation}
\label{Weyl2}
\hat{a}=\sqrt{\frac{1}{2\alpha\hbar}}\,\hat{q}+i\sqrt{\frac{\alpha}{2\hbar}}\,\hat{p}\ ,\quad
\hat{a}^\dag=\sqrt{\frac{1}{2\alpha\hbar}}\,\hat{q}-i\sqrt{\frac{\alpha}{2\hbar}}\,\hat{p}\ ,
\end{equation}
where $\alpha$ is a suitable parameter such that $\hat{a}^\#$ is a-dimensional and $[\hat{a}_i,\hat{a}^\dag_j]=\delta_{ij}$.
Then, one rewrites
\begin{equation}
\label{Weyl3}
\hat{W}_\hbar(r)=\exp{\Big(z_r\hat{a}^\dag-z^*_r\hat{a}}\Big)=:\hat{W}_\hbar(z_r)\ ,\quad z_r=-\frac{q}{\sqrt{2\alpha\hbar}}-i\,p\,
\sqrt{\frac{\alpha}{2\hbar}}\ .
\end{equation}
Let $\vert 0\rangle_\hbar$ denote the state annihilated by all operators $\hat{a}_j$: $\hat{a}_j\vert 0\rangle_\hbar=0$.
We shall refer to it as to the ground state, which in position representation amounts to the Gaussian state
\begin{equation}
\label{groundqp}
\langle q\vert 0\rangle_\hbar=\psi_0(q)=\frac{1}{(2\pi\alpha\hbar)^{s/4}}\,\exp{\Big(-\frac{\|q\|^2}{2\alpha\hbar}\Big)}\ .
\end{equation}
The coherent states
\begin{equation}
\label{Weyl4}
\vert z_{r}\rangle_\hbar=\hat{W}_\hbar(z_{r})\vert 0\rangle_\hbar={\rm e}^{-\|z_{r}\|^2/2}\, {\rm e}^{z_{r}\hat{a}^\dag}\,\vert 0\rangle_\hbar
\end{equation}
are eigenstates of the vector operator $\hat{a}$ with eigenvalue $z_r\in\CI^s$,
\begin{equation}
\label{coherentstate}
\hat{a}_j\vert z_r\rangle_\hbar=z_r^j\vert z_r\rangle_\hbar=\Big(\frac{q_j}{\sqrt{2\alpha\hbar}}+i\,p_j\,
\sqrt{\frac{\alpha}{2\hbar}}\Big)\vert z_r\rangle_\hbar\ ,
\end{equation}
whence
\begin{equation}
\label{Weyl5}
{}_\hbar\langle0\vert\hat{W}_\hbar(z_{r})\vert 0\rangle_\hbar={\rm e}^{-\|r\|_{\alpha,\hbar}^2}\
,\quad\hbox{where}\quad
\|r\|_{\alpha,\hbar}^2=\frac{1}{4\alpha\hbar}\|q\|^2+\frac{\alpha}{4\hbar}\|p\|^2\ .
\end{equation}

In order to set a useful algebraic setting for the classical limit, we will
consider the $C^*$ algebra of continuous functions over $\RI^{2s}$ which vanish
at infinity to which we add the identity function: we shall denote by $C_\infty$
this commutative $C^*$ algebra.
In this context, a particularly suitable algebraic setting for the classical
limit $\hbar\to0$ is the so-called \emph{anti-Wick} quantization that is based on
the over-completeness of coherent states:
\begin{equation}
\label{Weyl6}
\hat{1}=\frac{1}{(2\pi\hbar)^s}\int_{\RI^{2s}}{\rm d}r\, \vert z_{r}\rangle_\hbar{}_\hbar\langle z_{r}\vert\ ,
\end{equation}
where $\hat{1}$ denotes the identity operator on $\cH$.

Then, one may define two positive maps: a quantization map $\gamma_{\hbar,0}:
C_\infty(\RI^s)\mapsto\cW_\hbar$, given by
\begin{equation}
\label{Wick1}
C_\infty(\RI^{2s})\ni F\mapsto\gamma_{\hbar,0}[F]=:\hat{F}_\hbar \in \cW_\hbar\ ,\quad
\hat{F}_\hbar =\frac{1}{(2\pi\hbar)^s}\int_{\RI^{2s}}{\rm d}r\,F(r)\,
\vert z_{r}\rangle_\hbar{}_\hbar\langle z_{r}\vert\ ,
\end{equation}
which represents the quantization of the classical function $F(r)\in C_\infty(\RI^{2s})$,
and a de-quantization map
$\gamma_{0,\hbar}:\cW_\hbar\mapsto C_\infty(\RI^{2s})$ given by
\begin{equation}
\label{Wick2}
\cW_\hbar\ni \hat{X}\mapsto \gamma_{0,\hbar}[\hat{X}]\in C_\infty(\RI^{2s})\ ,\quad
X(r)={}_\hbar\langle z_{r}\vert \hat{X}\vert z_{r}\rangle_\hbar\ ,
\end{equation}
which \emph{de-quantizes} the operator $\hat{X}$ mapping it back to a function in $C_\infty(\RI^{2s})$.
\medskip

\begin{remark}
\label{rem0}
The quantization and de-quantization maps are positive as they
send positive functions into positive operators and vice versa;
they are unital as they map the identity function in $C_\infty(\RI^{2s})$ into
the identity operator $\hat{1}\in\cW_\hbar$.
\end{remark}
\medskip

Now, one computes
\begin{eqnarray}
\nonumber
\hskip -.5cm
\gamma_{0,\hbar}\circ\gamma_{\hbar,0}[F](r)&=&
\frac{1}{(2\pi\hbar)^s}\int_{\RI^{2s}}{\rm d}r'\, F(r')\,
\left|{}_\hbar\langle z_{r}\vert z_{r'}\rangle_\hbar\right|^2
=
\frac{1}{(2\pi\hbar)^s}\int_{\RI^{2s}}{\rm d}r'\, F(r')\,
\exp{\Big(-\|r'-r\|^2_{\alpha,\hbar}\Big)}\\
\label{Wick3}
&=&
\frac{1}{\pi^s}\int_{\RI^{s}\times\RI^{s}}{\rm d}u\,{\rm d}v\,
F\Big(q+u\sqrt{2\alpha\hbar},p+v\sqrt{\frac{2\hbar}{\alpha}}\Big)\,
\exp{\Big(-(\|u\|^2+\|v\|^2)\Big)}\ ,
\end{eqnarray}
whence the classical limit
\begin{equation}
\label{Wick4}
\lim_{\hbar\to0}\gamma_{0,\hbar}\circ\gamma_{\hbar,0}[F](r)=F(r)
\end{equation}
ensues.
If the classical system evolves in time according to a Hamiltonian
function $H(q,p)$ then the anti-Wick quantization allows one to recover
such an evolution from the quantized one when $\hbar\to0$,
the simplest situation occurs when $H(q,p)$ corresponds to a quantized
$\hat{H}=\sum_j\omega_j\hat{a}^\dag_j\hat{a}_j$. In such a case,
phase-space points $r=(q,p)$ evolve into $r_t=(q_t,p_t)=A_tr$
where $A_t$ is an $s\times s$ symplectic matrix; namely
\begin{equation}
\label{sympl2}
\Omega\, A_t\,=\,A^T_{-t}\,\Omega\ ,
\end{equation}
where $A_{-t}$, respectively $A^T_{-t}$ denote the inverse of the matrix $A_t$,
respectively its transposed.
Furthermore, exactly the same transformation affects the operators in $\hat{r}$
when subjected to the quantized Hamiltonian $\hat{H}$ while the
state $\vert 0\rangle_\hbar$ does not change. Therefore, Weyl operators are
sent into Weyl operators according to
\begin{eqnarray}
\nonumber
\hat{W}(r)\mapsto\cU_t[\hat{W}_\hbar(r)]&=&\hat{U}_t\,
\hat{W}_\hbar(r)\,\hat{U}^\dag_t={\rm e}^{it\hat{H}/\hbar}\,
\hat{W}_\hbar(r)\,{\rm e}^{-it\hat{H}/\hbar}\\
\label{Weyl7}
&=&\exp\Big(\frac{i}{\hbar}(r,\Omega\,A_t\hat{r})\Big)=
\exp\Big(\frac{i}{\hbar}(A_{-t}r,\Omega\,\hat{r})\Big)=
\hat{W}_\hbar(A_{-t}r)\ .
\end{eqnarray}
Then,
\begin{eqnarray}
\nonumber
{}_\hbar\langle z_{r}\vert \hat{U}_t\vert z_{r')}\rangle_\hbar&=&
{}_\hbar\langle z_{r}\vert\hat{U}_t
\hat{W}_\hbar(z_{r'})\hat{U}_t^\dag\vert 0\rangle_\hbar\\
\label{Weyl8}
&=&
{}_\hbar\langle z_{r}\vert z_{A_{-t}r'})\rangle_\hbar\ ,
\end{eqnarray}
so that
\begin{equation}
\label{Wick5}
\gamma_{0,\hbar}\circ\cU_t[\gamma_{\hbar,0}[F]](r)=
\int_{\RI^{2s}}{\rm d}r'\, F_t(r')\, \,
\exp{\Big(-\|r'-r\|^2_{\alpha,\hbar}\Big)}\ \quad
\hbox{where}\ F_t(r')=F(A_tr')\ .
\end{equation}
Then, in such a simple case, the classical limit of the quantum
time-evolution amounts to the classical time-evolution:
\begin{equation}
\label{timeclasslim}
\lim_{\hbar\to0}\gamma_{0,\hbar}\circ\cU_t[\gamma_{\hbar,0}[F]](r)=F_t(r)=F(A_tr)\ .
\end{equation}

\section{Noncommutative Quantum Mechanics}\label{sec3}

 We shortly review the formalism of noncommutative quantum mechanics,
 more details  being available in \cite{laure}.
 We consider the two dimensional noncommutative configuration space,
where the coordinates satisfy the commutation relation
\begin{equation}
 \left[ \hat x_i, \hat x_j\right] = i\theta\epsilon_{ij},
\end{equation}
 with $\theta$ a real positive parameter and $\epsilon_{i,j}$ the
completely antisymmetric tensor with $\epsilon_{1,2} = 1$.
Since, the operators
\begin{equation}
\label{Bop}
b = \frac{1}{\sqrt{2\theta}}(\hat x_1 + i\hat x_2)\ ,\quad
b^\dagger = \frac{1}{\sqrt{2\theta}}(\hat x_1 -i\hat x_2)
\end{equation}
satisfy the commutation relations $[b, b^\dagger] =  1$, one can introduce
a Fock-like vacuum vector $\vert 0\rangle$ such that
$b\vert 0\rangle=0$ and construct a non-commutative
configuration space isomorphic to the boson Fock space
\begin{equation}
 \mathcal{H}_c = \textrm{span}\{|n\rangle \equiv \frac{1}{\sqrt{n!}}
(b^\dagger)^n |0\rangle\}_{n=0}^{n = \infty},
\end{equation}
where the span is taken over the field of complex numbers.

A proper Hilbert space over such non-commutative configuration space
is the Hilbert-Schmidt Banach algebra $\mathcal{H}_q$ of bounded operators
$\psi(\hat x_1, \hat x_2)\in \mathcal{B}(\mathcal{H}_c)$ on $\mathcal{H}_c$ such that
\begin{equation}
\label{eq4}
tr_c(\psi(\hat x_1, \hat x_2)^\dagger\psi(\hat x_1, \hat x_2))< \infty \ .
\end{equation}
The $\textrm{tr}_c$ denotes the trace over non-commutative configuration space and
$\mathcal{B}(\mathcal{H}_c)$ the set of bounded operators on $\mathcal{H}_c$.
This space has a natural inner product and norm
\begin{equation}\label{eq5}
 (\phi(\hat x_1, \hat x_2),\psi(\hat x_1, \hat x_2)) =
\textrm{tr}_c(\phi(\hat x_1, \hat x_2)^\dagger\psi(\hat x_1, \hat x_2))
\end{equation}

Next we introduce the non-commutative Heisenberg algebra
\begin{eqnarray}
\label{NCHA}
 \left[\hat X_i,\:\hat P_j\right] = i\hbar\delta_{i,j},\quad
 \left[\hat X_i,\:\hat X_j\right] = i\theta\epsilon_{i,j},\quad
 \left[\hat P_i, \hat P_j\right] = 0,
 \end{eqnarray}
 where a unitary representation  in terms of the operators
$\hat X_i$ and $\hat P_i$ acting on the quantum Hilbert space (\ref{eq4}) with
the inner product (\ref{eq5}) is
\begin{eqnarray}
 \hat X_i\psi(\hat x_1, \hat x_2) = \hat x_i\psi(\hat x_1,\hat x_2), \quad
 \hat P_i\psi(\hat x_1,\hat x_2) = \frac{\hbar}{\theta}\epsilon_{i,j}
 \left[\hat x_j,\: \psi(\hat x_1,\hat x_2)\right].
\end{eqnarray}
In the above representation, the position acts by left multiplication and
the momentum adjointly.
We shall also consider the system to be equipped with a harmonic oscillator like Hamiltonian operator
\begin{eqnarray}
\label{e7}
 \hat H = \sum_{i =1}^2\left(\frac{1}{2m}\hat P_i^2 +
\frac{1}{2}m\omega^2\hat X_i^2 \right)\ ,
\end{eqnarray}
and refer to the model as two non-interacting non-commutative quantum oscillators.

One can associate to position and momentum operators creation and annihilation-like operators
$\hat A_{i}\,,\,\hat A^\dagger_{i}$, $i=1,2$ that satisfy the algebra
\begin{equation}\label{aacc}
\left[\hat A_{i},\: \hat A_{j}^\dagger \right]
= \delta_{ij};\:
\left[ \hat A_{i}, \: \hat A_{j} \right] = 0 \ .
\end{equation}
The explicit expressions of the $\hat A^\#_i$ are as follows~\cite{laure}
\begin{eqnarray}
\label{ca}
 \hat A_{1} &=& \frac{1}{\sqrt{K_+}}
\left( -\frac{\lambda_+}{\hbar}\hat X_1 - i\hat P_1 -
i\frac{\lambda_+}{\hbar}\hat X_2 +\hat P_2 \right)\ ,\quad
 \hat A_{1}^\dagger = \frac{1}{\sqrt{K_+}}
\left( -\frac{\lambda_+}{\hbar}\hat X_1 + i\hat P_1
+i\frac{\lambda_+}{\hbar}\hat X_2 +\hat P_2 \right)\\
\label{caa}
\hat A_{2} &=& \frac{1}{\sqrt{K_-}}
\left( \frac{\lambda_-}{\hbar}\hat X_1 + i\hat P_1
-i\frac{\lambda_-}{\hbar}\hat X_2 +\hat P_2 \right)\ ,\quad
\hat A_{2}^\dagger = \frac{1}{\sqrt{K_-}}
\left( \frac{\lambda_-}{\hbar}\hat X_1 - i\hat P_1 +
i\frac{\lambda_-}{\hbar}\hat X_2 +\hat P_2 \right)\ ,
\end{eqnarray}
where
\begin{eqnarray}
\label{lambdas}
 \lambda_\pm = \frac{1}{2}
\left( m\omega\sqrt{4\hbar^2 + m^2\omega^2\theta^2} \pm m^2\omega^2\theta\right),\quad
K_\pm =& \lambda_\pm\left(4 \pm \frac{2\lambda_\pm\theta}{\hbar^2} \right).
\end{eqnarray}
Interestingly, the operators $\hat A^\#_j$ can be interpreted as proper
annihilation and creation operators as there is a vector in $\mathcal{H}_q$, that is
a Hilbert-Schmidt operator $\psi_{0}$ such that
\begin{eqnarray}
\hat A_{1}\vert\psi_{0}\rangle = \hat A_2\vert\psi_0\rangle=0 \ ,
\end{eqnarray}
given by~\cite{laure}
\begin{equation}
\label{NNground}
\psi_{0}(\hat x_1, \hat x_2)
 = \exp{\Big(\frac{\beta}{2\theta}(\hat x_1^2 +\hat x_2^2)\Big)}\ ,\quad \beta =
 \ln( 1-\frac{\theta}{\hbar^2}\lambda_-)
= -\ln(1 +\frac{\theta}{\hbar^2}\lambda_+).
\end{equation}
After normalization, the ground state corresponding to $\vert\psi_0\vert$ is
\begin{equation}
\label{ground}
|0,0\rangle = \frac{\vert\psi_0\rangle}{\sqrt{\mathcal{N}}}\ , \quad
\mathcal{N} = \frac{\hbar^4}{2\hbar^2\lambda_--\theta\lambda^2_-}\ .
\end{equation}
Furthermore, the Hamiltonian (\ref{e7}) becomes
\begin{equation}\label{e8}
 \hat H = \frac{\lambda_+}{m}
 \hat A_{1}^\dagger \hat A_{1}
 +\frac{\lambda_-}{m}
 \hat A_{2}^\dagger \hat A_{2}
 +\frac{\lambda_{+} +\lambda_{-}}{2m}
\end{equation}

Clearly, there are two possible quantization and de-quantization schemes
playing possibly together in this context: one is passing from a commutative
to a non-commutative configuration space and back, another one is to pass from
commuting position and momentum operators to non-commuting ones and back.
In order to make the anti-Wick quantization  works, we proceed by extending
the coherent state construction of the previous section to this  non-commutative
quantum system with two degrees of freedom.

\subsection{Gaussian-like states of the non-commutative
quantum harmonic oscillators}
\label{sec4}

In analogy with what we presented in Section~(\ref{sec2}), we introduce
the coordinate vector $r = (x_1,x_2,y_1,y_2)$ and the operator vector
$\hat r =\left(\hat X_1,\hat X_2, \hat P_1, \hat P_2\right)$.
Then, we construct the Weyl-like operators
\begin{equation}
\label{w1}
\hat W_{\hbar,\theta}(r) = \exp{\Big(\frac{i}{\mu_{\hbar,\theta}}
\left(r,\Omega\hat r\right)\Big)}\ ,
\end{equation}
where $\mu_{\hbar,\theta}$ is a parameter with the dimension of an action.
Using the commutation relations~(\ref{NCHA}), the Weyl algebraic composition
law (\ref{Weyl1}) now read
\begin{equation}
\label{Weylalgnc}
\hat W_{\hbar,\theta}(r)\hat W_{\hbar,\theta}(r')
 = \exp{\left(-\frac{i(\hbar+\theta)\mu^2_{\hbar,\theta}}{2}\,(r,\Omega r')\right)} \,
 \hat W_{\hbar,\theta}(r+r');\quad \Omega' =
 \left(
 \begin{array}{cccc}
 0 & 0 & 1 & 0\\
 0 & 0 & 0 & 1\\
 -1 & 0 & 0 & 0 \\
 0 & -1 & 0 & 0
 \end{array}\right)
\end{equation}

As to the parameter $\mu_{\hbar,\theta}$, it will eventually
let vanish with $\hbar\to0$ and $\theta\to0$.
However, there are three possible ways we can reach the full
commutative limit $\hbar= 0 =\theta $:
\begin{enumerate}
\item
by linking $\hbar$ and $\theta$ so that one may consider the
classical limit $\mu_{\hbar,\theta}\to0$;
\item
by letting $\theta\to0$ first so to get to standard quantum mechanics
and then let $\hbar\to0$;
\item
by letting $\hbar\to0$ first so to get to a non-quantum non-commutative
system and then let $\theta\to0$.
\end{enumerate}
In order to explore these three possibilities, we shall choose $\mu_{\hbar,\theta}$ such that
\begin{equation}\label{para}
\lim_{\theta \rightarrow 0}\mu_{\hbar,\theta} = \hbar; \quad
\lim_{\hbar \rightarrow 0}\mu_{\hbar,\theta}=m\omega\theta\ .
\end{equation}
Notice that the latter expression is the only natural constant with
the dimensions of an action when $\hbar=0$ in the model.
A most natural choice is provided by (\ref{lambdas})
\begin{equation}
\label{choice}
\mu_{\hbar,\theta}= \frac{\lambda_+}{m\omega} =
\frac{\sqrt{4\hbar^2 + m^2\omega^2\theta^2} + m\omega\theta}{2}\ ,
\end{equation}
whereas $\lambda_{-}\to0$ when $\hbar\to0$.

By inverting the relations~(\ref{ca}) and~(\ref{caa}),
\begin{eqnarray}
 \hat X_1 &=& \frac{\hbar}{2(\lambda_+ + \lambda_{-})}
 \left(\sqrt{K_{-}}(A_2 + A_2^\dagger) - \sqrt{K_+}(A_1 +A_1^\dagger)\right)\\
 \hat X_2 &=& -\frac{\hbar}{2i(\lambda_+ + \lambda_{-})}
 \left(\sqrt{K_+}(A_1 -A_1^\dagger) + \sqrt{K_{-}}(A_2 -A_2^\dagger)\right)\\
 \hat P_1 &=& \frac{1}{2i(\lambda_+ + \lambda_{-})}
 \left(\lambda_+\sqrt{K_{-}}(A_2 - A_2^\dagger) - \lambda_{-}\sqrt{K_+}(A_1 -A_1^\dagger)\right)\\
 \hat P_2 &=& \frac{1}{2(\lambda_+ + \lambda_{-})}
 \left(\lambda_+\sqrt{K_{-}}(A_2 + A_2^\dagger) + \lambda_{-}\sqrt{K_+}(A_1 +A_1^\dagger)\right)\ ,
\end{eqnarray}
one can rewrite the Weyl operators~(\ref{w1}) in the form
\begin{eqnarray}
\label{w2}
\hat W_{\hbar,\theta}(z_r) = \exp{\Big(z_{1,r}\hat A^\dag_1\,
+\,z_{2,r}\hat A^\dag_2\,-\,z^*_{1,r}\hat A_1\,-\,z^*_{2,r}\hat A_2\Big)}\ ,
\end{eqnarray}
which is similar to~(\ref{Weyl3}), with $z_r=(z_{1,r},z_{2,r})$ a
two dimensional complex vector whose real and imaginary parts are
connected to the real four dimensional vector $r$ by
\begin{equation}
\label{XZ}
\begin{pmatrix}
\mathcal{R}e(z_{1,r})\cr
\mathcal{R}e(z_{2,r})\cr
\mathcal{I}m(z_{1,r})\cr
\mathcal{I}m(z_{2,r})
\end{pmatrix}= \hat{J} r;\quad
\hat{J} = \frac{1}{2\mu_{\hbar,\theta}(\lambda_{+} +\lambda_{-})}
\begin{pmatrix}
 \lambda_{-}\sqrt{K_+} & 0 & 0 & -\hbar\sqrt{K_+}\cr
 -\lambda_+\sqrt{K_{-}} & 0 & 0 & -\hbar\sqrt{K_{-}}\cr
 0 & \lambda_{-}\sqrt{K_+} & \hbar\sqrt{K_+} & 0\cr
 0 & \lambda_{+}\sqrt{K_{-}} & -\hbar\sqrt{K_{-}} & 0
\end{pmatrix}\ .
\end{equation}
By using the ground state~(\ref{ground}) and the relations~(\ref{aacc}),
we now introduce the non-commutative analogues of the coherent states~(\ref{Weyl4}),
\begin{equation}
 \vert z_r\rangle_{\hbar,\theta} = \hat W_{\hbar,\theta}(z_r)\vert 0,0 \rangle\,
 =\,\exp{\Big(-\frac{\|z_r\|^2}{2}\Big)}\, \exp{\Big(z_{1,r}\hat A_1^\dag+z_{2,r}
 \hat A^\dag_2\Big)}\,\vert0,0\rangle\ ,
\end{equation}
where $\|z_r\|^2=|z_{1,r}|^2+|z_{2,r}|^2$.
Exactly as in the case of~(\ref{coherentstate}), because of the algebraic
relations~(\ref{aacc}), it follows that
\begin{equation}
\label{cohNC}
\hat A_1\,\vert z_r\rangle_{\hbar,\theta}\,=\,z_{1,r}\,\vert z_r\rangle_{\hbar,\theta}\
,\quad \hat A_2\vert z_r\rangle_{\hbar,\theta}\,=
\,z_{2,r}\,\vert z_r\rangle_{\hbar,\theta}\ .
\end{equation}
These states are not exactly coherent states as they do not satisfy
the non-commutative analog of minimal indeterminacy~\cite{joe}; however,
they have a Gaussian character and constitute an over-complete set.

\begin{lemma}
The states $\vert z_r\rangle$  satisfy the resolution of identity
\begin{equation}
\frac{1}{\pi^2}\int_{\CI^2} {\rm d}z_r\, \vert z_r\rangle_{\hbar,\theta}{}_{\hbar,\theta}\langle z_r\vert
= \frac{J}{\pi^2}\int_{\RI^4}\,{\rm d}r
\, \vert z_r\rangle_{\hbar,\theta}{}_{\hbar,\theta}
\langle z_r\vert\ = \hat{1},
\end{equation}
where $\displaystyle J=\textrm{Det}\hat{J} =
\frac{\hbar^2}{4\mu_{\hbar,\theta}^4}$ with $\hat J$ the transformation matrix in~(\ref{XZ}).
\end{lemma}

\begin{proof}
Denote the integral by $\hat{I}$; then,
one checks whether $\langle n_1,n_2\vert\hat I\vert m_1,m_2\rangle=\delta_{n_1,m_1}\delta_{n_2,m_2}$,
where the states
$$
\vert n_1,n_2\rangle=\frac{(\hat A_1^\dag)^{n_1}(\hat A_2^\dag)^{n_2}}{\sqrt{n_1!n_2!}}\,\vert 0,0\rangle
$$
constitute an orthonormal basis in the non-commutative Hilbert space $\mathcal{H}_q$.
Then,~(\ref{cohNC}) yields
$$
\langle n_1,n_2\vert\hat{I}\vert m_1,m_2\rangle=
(z_{1,r}^*)^{n_1}(z_{2,r}^*)^{n_2}(z_{1,r})^{m_1}(z_{2,r})^{m_2}\,e^{-\|z_r\|^2}\ ,
$$
whence the result follows by Gaussian integration.
\end{proof}

\section{The classical limits of the 
non-commutative harmonic oscillators}
\label{sec5}

Following the prescriptions of the anti-Wick 
quantization in Section \ref{sec2}, we start 
by choosing the classical algebra, that we choose as $C_\infty(\RI^4)$ made of continuous functions that vanish at infinity augmented with the identity function.
Then, following~(\ref{Wick1}) and~(\ref{Wick2}), we define the quantization map de-quantization maps.

\begin{definition}
Let $\cW_{\hbar,\theta}$ be the $C^*$ algebra generated by the Weyl operators~(\ref{w1}),
the quantization of $F\in C_\infty(\RI^4)$
will be given by the positive unital map
$\gamma_{(\hbar,\theta),0}:C_\infty(\RI^4)\mapsto \cW_{\hbar,\theta}$ defined by
\begin{equation}
\label{qNC}
C_\infty(\RI^4)\ni F\mapsto\gamma_{(\hbar,\theta),0}[F]=:\hat F_{\hbar,\theta}\in\cW_{\hbar,\theta}\ ,\quad
\hat F_{\hbar,\theta}=\frac{J}{\pi^2}\int_{\RI^4} {\rm d}r\, F(r)\,
\vert z_r\rangle_{\hbar,\theta}{}_{\hbar,\theta}\langle z_r\vert\ ,
\end{equation}
while the de-quantization map by the following positive, unital
map $\gamma_{0,(\hbar,\theta)}:\cW_{\hbar,\theta}\mapsto C_\infty(\RI^4)$
\begin{equation}
\label{dqNC}
\cW_{\hbar,\theta}\ni\hat X\mapsto
\gamma_{0,(\hbar,\theta)}[\hat X]=:X(r)\in C_\infty(\RI^4)\ ,
\quad X(r)={}_{\hbar,\theta}\langle z_r\vert\,\hat X\,\vert z_r\rangle_{\hbar,\theta}\ .
\end{equation}
\end{definition}
\medskip

In order to study the classical limit of the non-commutative quantum oscillators we shall focus upon the following
functions
\begin{equation}
\label{qdqNC}
C_\infty(\RI^4)\ni F\mapsto
F_{\hbar,\theta}=\gamma_{0,(\hbar,\theta)}\circ\gamma_{(\hbar,\theta),0}[F]\in C_\infty(\RI^4)
\end{equation}
that, after some manipulations reported in the Appendix, explicitly reads
\begin{eqnarray}
\label{qdqNC1}
F_{\hbar,\theta}(r) &=& \frac{J}{\pi^2}
\int_{\mathbb{R}^4} {\rm d}r'\, F(r')\,\left\vert \langle z_r\vert z_{r'}\rangle\right\vert^2\\
\nonumber
&=&\frac{1}{\pi^2}
\int_{\mathbb{R}^4}{\rm d}w\,
e^{-\|w\|^2}\,F\Big(x_1+f(w_1,w_2)\ ,\ x_2+f(w_3,w_4)\ ,\ y_1+g(w_3,w_4)\ ,\ y_2-g(w_1,w_2)\Big)\ ,
\label{qdqNC2}
\end{eqnarray}
where
\begin{eqnarray}
\label{f12}
&&
\hskip -1cm
f(x,y)=\frac{\mu_{\hbar,\theta}\sqrt[4]{4\hbar^2 + m^2\omega^2\theta^2}}{2\sqrt{m\omega}\hbar}
 \Big(\frac{x}{\sqrt{\gamma_+}}+\frac{y}{\sqrt{\gamma_{-}}}\Big)\ ,\
g(x,y)=\frac{\mu_{\hbar,\theta}\sqrt{m\omega}\sqrt[4]{4\hbar^2 + m^2\omega^2
\theta^2}}{\sqrt{4\hbar^2 + 2m^2\omega^2\theta^2}}\Big(\frac{x}{\sqrt{\gamma_+}}-
\frac{y}{\sqrt{\gamma_{-}}}\Big)\ ,\\
\label{g12}
&&
\gamma_\pm = \frac{1}{2}\Big(1 \pm\frac{m\omega\theta}
{\sqrt{4\hbar^2 + 2m^2\omega^2\theta^2}}\Big)\ .
\end{eqnarray}

In the following, we compute and discuss various possible limits
in terms of $\hbar$ and $\theta$ or both going to zero.

\subsection{Classical limit: $\mu_{\hbar,\theta}\to0$}
\label{secNCQ2C}

If $\hbar$ and $\theta$ vanish together with the same speed,
that is if $\hbar=\alpha\,\theta$, with $\alpha$ a suitable constant, then
$\mu_{\hbar,\theta}\simeq \hbar$, $\gamma_\pm$ tend to constants
and we get the classical limit
\begin{equation}
\lim_{\hbar\to0}F_{\hbar,\theta}(r)
= \frac{1}{\pi^2}\int_{\mathbb{R}^4}{\rm d}w\,
 e^{-\|w\|^2}\, F(r)=  F(r)\ .
\label{NCQ2C}
\end{equation}
\subsection{Commutative configuration-space limit: $\theta\to0$}
\label{NCQ2Q}
In the limit $\theta\to0$, $\mu_{\hbar,\theta}\to \hbar $, from~(\ref{f12}) and~(\ref{g12}) we get the limit behaviours
$$\gamma_\pm=\frac{1}{2}\ ,\quad f(x,y)=\sqrt{\frac{\hbar}{m\omega}}(x+y)\ ,
\quad g(x,y)=\sqrt{\hbar m\omega}(x-y)\ ,$$
so that
\begin{eqnarray}
\nonumber
 &&
F_\hbar(r) = \lim_{\theta\to0}F_{\hbar,\theta}(r)
= \frac{1}{\pi^2}\int_{\mathbb{R}^4}{\rm d}w\,
 e^{-\|w\|^2}\, \times\\
 \label{wickl0}
&&
\hskip-.5cm\times\,
F\left(x_1 +\sqrt{\frac{\hbar}{m\omega}}(w_1 + w_2),
 x_2 +\sqrt{\frac{\hbar}{m\omega}}(w_3 +w_4),y_1 +\sqrt{\hbar m\omega}(w_3 -w_4),
 y_2 +\sqrt{\hbar m\omega}(w_2 -w_1)\right)\ .
\label{NCQ2Qa}
\end{eqnarray}
This is nothing but the map~(\ref{Wick3}) for two independent harmonic oscillators
with $\alpha=(m\omega)^{-1}$, in fact by a change of variable that we include
in the Appendix, we show that the equation (\ref{NCQ2Qa}) is equivalent
to
\begin{eqnarray}
\nonumber
 &&
F_\hbar(r)=
\frac{1}{\pi^2}\int_{\mathbb{R}^4}{\rm d}u{\rm d}v\,
 e^{-u ^2- v^2}\, \times\\
\label{wickl}
&&
\times\,
F\left(x_1 +\sqrt{\frac{2\hbar}{m\omega}}u_1,
 x_2 +\sqrt{\frac{2\hbar}{m\omega}}u_2,y_1 +\sqrt{2\hbar m\omega}v_1,
 y_2 +\sqrt{2\hbar m\omega}v_2\right)\ .
\end{eqnarray}

The corresponding Weyl operators are
$$
\hat W_\hbar(r)=\exp{\Big(\frac{i}{\hbar}(r,\Omega\hat r)\Big)}\ ,
$$
and the Gaussian ground state
$$
\psi_0(x_1,x_2)=\sqrt[4]{\frac{m\omega}{\pi\hbar}}\, \exp{\Big(-\frac{m\omega}{2\hbar}(x_1^2+x_2^2)\Big)}\ ,
$$
in the $\hat x_{1,2}$ position representation.
The classical limit $\hbar\to0$ then yields
\begin{eqnarray}
\label{NCQ2Q2C}
\lim_{\hbar\rightarrow 0}F_\hbar(r)=F(r)\ ,
\end{eqnarray}
exactly as in the previous  Section.

\subsection{Non-Commutative configuration-space limit: $\hbar\to0$}
\label{NCQ2NC}

In the limit $\hbar\to0$, $\mu_{\hbar,\theta}\to m\omega\theta$,
from~(\ref{f12}) and~(\ref{f12}) we get the limit behaviours
\begin{eqnarray*}
\gamma_\pm&=&\frac{1}{2\sqrt{2}}(\sqrt{2}\pm1)\ ,\quad f(x,y)\to+\infty\quad\hbox{almost everywhere on $\RI^2$}\\
g(x,y)&=& m\omega\sqrt{\theta}\left(\frac{x}{\sqrt{1+\frac{1}{\sqrt{2}}}}-\frac{y}{\sqrt{1-\frac{1}{\sqrt{2}}}}\right)\ ,
\end{eqnarray*}
so that
\begin{eqnarray}
\nonumber
&&
F_\theta(r) = \lim_{\hbar\rightarrow 0}F_{\hbar,\theta}(r)
 = \frac{1}{\pi^2}\int_{\mathbb{R}^4}{\rm d}w\,
 e^{-\|w\|^2}\,\times\\
\label{NCQ2NCa}
&&\hskip .5cm
\times\,F_\infty\left(y_1+ m\omega\sqrt{\theta}
 \left(\frac{w_3}{\sqrt{1+\frac{1}{\sqrt{2}}}} -
 \frac{w_4}{\sqrt{1-\frac{1}{\sqrt{2}}}}\right),
 y_2 + m\omega\sqrt{\theta}
 \left(\frac{w_2}{\sqrt{1-\frac{1}{\sqrt{2}}}} -
 \frac{w_1}{\sqrt{1+\frac{1}{\sqrt{2}}}}\right)\right)\  ,
 \end{eqnarray}
where the function $F_\infty(y_1,y_2)$ denotes the limit
$\displaystyle \lim_{x_{1,2}\to+\infty}F(r)$. Such limit exists and it is not trivial, in general,
because the $C^*$ algebra $C_\infty(\RI^4)$ contains also functions of the form
$$
(f_1(x_1)+c_1)\,(f_2(x_2)+c_2)\,g_1(y_1)\,g_2(y_2)\ ,
$$
where $c_i$ are constants and $f_i(x_i)$, $i=1,2$, vanish when their arguments go to $\pm\infty$.

The resulting expression coincides with the map~(\ref{Wick3}) for the case of a commutative $C^*$ algebra of functions $C_\infty(\RI^2)$,
Weyl operators of the form
\begin{equation}
\label{WeylNC}
\hat W_\theta(r)=\exp{\Big(\frac{i}{m\omega\theta}(y_1\hat x_2-y_2\hat x_1)\Big)}\ ,\qquad [\hat x_1\,,\,\hat x_2]=i\theta\ , r=(y_1,y_2)
\end{equation}
and Gaussian ground state
$$
\psi_0(x_1)=\sqrt[4]{\frac{1}{\pi\theta}}\, \exp{\Big(-\frac{x_1^2}{2\theta}\Big)}\ ,
$$
in the $\hat x_1$-representation where $(\hat x_2\psi)(x_1)=-i\theta\psi'(x_1)$.
Under suitable change of variables, (\ref{NCQ2NCa}) 
is equivalent to
\begin{equation}\label{NCQ2NCa1}
F_\theta(y_1,y_2)=
\frac{1}{\pi}\int_{\mathbb{R}\times\mathbb{R}}dv\;e^{-\|v\|^2}\,
F_\infty(y_1+m\omega\sqrt{2\theta}v_1,\: y_2 + m\omega\sqrt{2\theta}v_2)\ .
\end{equation}
This is the expression ~(\ref{Wick3}) with $\hbar$ substituted by $m\omega\theta$
and $\alpha = (m\omega)^{-1}$.

Then, in analogy with Section \ref{sec2}, one defines two positive maps.
The first map is a configuration space quantization map $\gamma_{\theta,0}: C_\infty(\RI^2)\mapsto\cW_\theta$ from the $C^*$ algebra of continuous functions over $\RI^2$
which vanish at infinity equipped with the identity function into the $C^*$ algebra generated by the Weyl operators~(\ref{WeylNC}),
\begin{equation}
\label{Wick1NCS}
C_\infty(\RI^{2})\ni F\mapsto\gamma_{\theta,0}[F]=:\hat{F}_\theta \in \cW_\theta\ ,\quad
\hat{F}_\theta =\frac{1}{2\pi m\omega\theta}\int_{\RI^{2}}{\rm d}r\,F(r)\,
\vert z_{r}\rangle_\theta{}_\theta\langle z_{r}\vert\ ,
\end{equation}
where
\begin{equation}
\label{cohNCS}
\vert z_{r}\rangle_\theta=\exp{\Big(\frac{i}{m\omega\theta}(y_1\hat x_2-y_2\hat x_1)\Big)}\vert\psi_0\rangle_\theta\ ,\quad
z_r=-\frac{1}{m\omega\sqrt{2\theta}}(y_1+iy_2)\ .
\end{equation}
The second map is a de-quantizing configuration space map  $\gamma_{0,\theta}:\cW_\theta\mapsto C_\infty(\RI^{2})$ given by
\begin{equation}
\label{Wick2NCS}
\cW_\theta\ni \hat{X}\mapsto \gamma_{0,\theta}[\hat{X}]\in C_\infty(\RI^{2})\ ,\quad
X(r)={}_\theta\langle z_{r}\vert \hat{X}\vert z_{r}\rangle_\theta\ ,
\end{equation}
which \emph{de-quantizes} the operator $\hat{X}$ mapping it back to a function in $C_\infty(\RI^{2})$.
By combining the two maps, one finds that $\gamma_{0,\theta}\circ\gamma_{\theta,0}[F](r)$ equals~(\ref{NCQ2NCa1}).

By letting $\theta\to 0$, one removes the non-commutativity of the configuration space and get back to a continuous function, on $\RI^2$ instead of
$\RI^4$:
\begin{equation}
\label{NCQ2NC2C}
\lim_{\theta\rightarrow 0}F_\theta(r)= F_\infty(y_1,y_2)\ .
\end{equation}

We thus see that removal of quantum non-commutativity followed by removal of configuration space non-commutativity does not get back to the initial commutative algebra of continuous functions over $\RI^4$, but on "half" space.
Therefore, the two de-quantizing limits do not commute:
\begin{equation}
\label{NCL}
\lim_{\theta\to0}\lim_{\hbar\to0}\,\neq\,\lim_{\hbar\to0}\lim_{\theta\to0}\ .
\end{equation}
In the next section we study how this non-exchangeability of limits affects as simple a time-evolution as the one generated by the Hamiltonian~(\ref{e7}).

\section{Classical limit of the non-commutative time evolution}
\label{sec6}

We now consider the time-evolution generated by the Hamiltonian~(\ref{e7}), using as dimensional action, not $\hbar$,
but the parameter$\mu_{\hbar,\theta}$ in~(\ref{choice}).
The unitary time-evolutor on the non-commutative Hilbert space $\mathcal{H}_q$ is thus given by
\begin{equation}
\label{time-evolutor}
\hat U_t = \exp{\Big(-\frac{it}{\mu_{\hbar,\theta}}\hat H\Big)}\ .
\end{equation}
Its action on the Weyl operators in the forms~(\ref{w1}) and~(\ref{w2}) is easily computed to be
\begin{eqnarray}
\nonumber
\hat U^\dag_t\,\hat W_{\hbar,\theta}(z)\,\hat U_t&=&\exp{\Bigg(e^{\frac{it}{m\mu_{\hbar,\theta} }\lambda_+}z_1\hat A_1^\dagger+
e^{\frac{it}{m\mu_{\hbar,\theta} }\lambda_-}z_2\hat A_2^\dagger\,-\,e^{-\frac{it}{m\mu_{\hbar,\theta} }\lambda_+}\bar{z}_1\hat A_1
\,-\,e^{-\frac{it}{m\mu_{\hbar,\theta} }\lambda_-}\bar{z}_2\hat A_2\Bigg)}\\
\label{Weylt}
&=& \exp{\Big(r,\Omega A_{t,\hbar,\theta}\hat r\Big)}=
\exp{\Big(A_{-t,\hbar,\theta}r,\Omega\hat r\Big)}=
\hat W_{\hbar,\theta} (r_{-t})\ ,
\end{eqnarray}
where, from symplecticity,
$\Omega A_{t,\hbar,\theta}=A^T_{-t,\hbar,\theta}\Omega$,
and then
\begin{equation}
\label{Weylt1}
r_{-t} = A_{-t,\hbar,\theta}\, r\ , \quad
 A_{t,\hbar,\theta} = \left(
 \begin{array}{cccc}
  \cos\omega_+t & 0 & -\sin\omega_+t & 0\\
  0 & \cos\omega_{-}t & 0 & -\sin\omega_{-}t\\
  \sin\omega_+t & 0 & \cos\omega_+t & 0\\
  0 & \sin\omega_{-}t & 0& \cos\omega_{-}t
 \end{array}
 \right)\ ,
\end{equation}
with the oscillation frequencies given by
\begin{equation}
\label{frequencies}
\omega_\pm =
\frac{\lambda_\pm}{m\mu_{\hbar,\theta}}=\frac{\left( m\omega\sqrt{4\hbar^2 + m^2\omega^2\theta^2} \pm m^2\omega^2\theta\right)}{2m\mu_{\hbar,\theta}}\ .
\end{equation}
Since the ground state $\vert 0,0\rangle$ in~(\ref{ground}) is left invariant by $U_t$, one finds that the time-evolution of the quantized function
in~(\ref{qNC}) is given by
\begin{equation}
\label{qNCt}
\hat F_{\hbar,\theta}(t)=\hat U_t^\dag\,F_{\hbar,\theta}\, \hat U_t\in\cW_{\hbar,\theta}=
\frac{J}{\pi^2}\int_{\RI^4} {\rm d}r\, F(r)\, \vert z_r(-t)\rangle_{\hbar,\theta}{}_{\hbar,\theta}\langle
z_r(-t)\vert=
\frac{J}{\pi^2}\int_{\RI^4} {\rm d}r\, F_t(r)\,
\vert z_r\rangle_{\hbar,\theta}{}_{\hbar,\theta}\langle z_r\vert\ ,
\end{equation}
where it has been used that $Det(A_{t,\hbar,\theta})=1$ and has been set $F_t(r)=F(A_{t,\hbar,\theta}\,r)$.
Then,~(\ref{qdqNC}) yields
\begin{eqnarray}
\nonumber
&&
F_{\hbar,\theta,t}(r)=\gamma_{0,(\hbar,\theta)}\Big[\hat U_t^\dag\gamma_{(\hbar,\theta),0}[F]\hat U_t\Big](r)
= \frac{J}{\pi^2}
\int_{\mathbb{R}^4} {\rm d}r'\, F_t(r')\,\left\vert \langle z_r\vert z_{r'}\rangle\right\vert^2\\
\label{qdqNCt}
&&
=\frac{1}{\pi^2}
\int_{\mathbb{R}^4}{\rm d}w\,
e^{-\|w\|^2}\,F_t\Big(r+h(w)\Big)\ , \quad h(w)=\Big(f(w_1,w_2)\ ,f(w_3,w_4)\ ,g(w_3,w_4)\ ,-g(w_1,w_2)\Big)\ ,
\end{eqnarray}
with the functions $f,g$ as in~(\ref{f12}).
\subsection{Classical limit: $\mu_{\hbar,\theta}\to0$}
\label{NCQ2Cc}

If $\hbar$ and $\theta$ vanish together with the same speed,
that is if $\hbar=\alpha\,\theta$, with $\alpha$ a suitable constant, then
$\mu_{\hbar,\theta}\simeq \hbar$, $\gamma_\pm$ tend to constants
and we get the classical limit as in the time independent case
\begin{equation}
\lim_{\hbar\to0}F_{t,\hbar,\theta}(r)
= \frac{1}{\pi^2}\int_{\mathbb{R}^4}{\rm d}w\,
 e^{-\|w\|^2}\, F_t(r)=  F_t(r)\ .
\label{NCQ2C1}
\end{equation}

\subsection{Commutative configuration-space limit: $\theta\to 0$}

By letting $\theta\to 0$ in~(\ref{qdqNCt}) and thus recovering the commutative quantum mechanics context, from~(\ref{frequencies}) one has
$\lim_{\theta\rightarrow 0}\omega_\pm = \omega $ while for the evolution matrix in~(\ref{Weylt})
\begin{equation}
\label{At0}
 A_t=\lim_{\theta\to0} A_{t,\hbar,\theta}=\left(
 \begin{array}{cccc}
  \cos\omega t & 0 & -\sin\omega t & 0\\
  0 & \cos\omega t & 0 & -\sin\omega t\\
  \sin\omega t & 0 & \cos\omega t & 0\\
  0 & \sin\omega t & 0& \cos\omega t
 \end{array}
 \right)\ .
\end{equation}
Then, in this limit one gets
\begin{equation}
F_{t,\hbar}(r)= \lim_{\theta\rightarrow 0}F_{t,\hbar,\theta}(r) =
 \frac{1}{\pi^2}
 \int_{\mathbb{R}^4}dw\, e^{-||w||^2}\, F\Big( A_{-t}\,(r + h(w))\Big) \ .
\end{equation}
This corresponds to the commutative quantum mechanical time evolution of two identical identical independent harmonic oscillators.

In the classical limit $\hbar\to0$ one obviously recovers the time-evolution of two classical harmonic oscillators
whose canonical coordinates evolve according to the symplectic matrix (\ref{At0}):
\begin{equation}
\label{clt}
 \lim_{\hbar\rightarrow 0}F_{t,\hbar}(r)= F( A_{-t}\, r)\ .
\end{equation}

\subsection{Non-commutative configuration-space limit: 
$\hbar \to 0$}

By letting $\hbar\to0$ in~(\ref{qdqNCt}) and thus going to the non-commutative configuration space context, from~(\ref{frequencies}) one has
$\lim_{\theta\rightarrow 0}\omega_+ = \omega$,  while $\lim_{\theta\rightarrow 0}\omega_- = 0$; thus, for the evolution matrix in~(\ref{Weylt})
\begin{equation}
 B_t=\lim_{\hbar\to0} A_{t,\hbar,\theta}=\left(
 \begin{array}{cccc}
  \cos\omega t & 0 & -\sin\omega t & 0\\
  0 & 1 & 0 & 0\\
  \sin\omega t & 0 & \cos\omega t & 0\\
  0 & 0  & 0& 1
 \end{array}
 \right)\ .
\end{equation}
The previous matrix cannot be used directly in performing the limit in in~(\ref{qdqNCt}); indeed,
we have to take into account that $A_{t,\hbar,\theta}$ mixes the components of the vector $r+h(w)$ and
the function $f(x,y)$ diverges as $1/\hbar$.
However, when $f(x,y)$ multiplies $\sin\omega_-t$ the product vanishes since $\omega_-\simeq \hbar^2$.
Therefore, when $\hbar \to 0$, from~(\ref{NCQ2NCa}), one gets
\begin{eqnarray}
\nonumber
F_{t,\theta}(y_2)=\lim_{\hbar\rightarrow 0}F_{t,\hbar,\theta}(r)
&=& \frac{1}{\pi^2}\int_{\mathbb{R}^4}{\rm d}w\,
 e^{-\|w\|^2}\,F_\infty\left(y_2 + m\omega\sqrt{\theta}
 \left(\frac{w_2}{\sqrt{1-\frac{1}{\sqrt{2}}}} -
 \frac{w_1}{\sqrt{1+\frac{1}{\sqrt{2}}}}\right)\right)\\
\label{NCQ2NCat}
&=& \frac{1}{\sqrt{\pi}}\int_{\mathbb{R}}dv_2\;
  e^{-v_2^2}F_\infty(y_2 +m\omega\sqrt{2\theta}v_2)\ ,
\end{eqnarray}
where the function $F_\infty(y_2)$ denotes the limit
$\displaystyle \lim_{x_1,x_2,y_1\to+\infty}F(r)$ and is effectively a function of $y_{1,2}$ only while $r=(x_1,x_2,y_1,y_2)$.
The only footprint of the non-commutative quantum dynamics is the reduction of the dependence of the initial continuous functions from $r\in\RI^4$ ro $y_2\in\RI$. Indeed, the full classical limit yields
\begin{equation}
 \lim_{\theta\to 0}F_{t,\theta}(y_2)= F_\infty({y_2})\ .
\end{equation}
Therefore, starting with the continuous functions over $\RI^4$, letting the dynamics act and then removing the standard non-commutativity before removing the configuration space non-commutativity one loses track of the time-evolution and even reduces, after the complete classical limit, the domain of definition of the continuous functions from $\RI^4$ to $\RI$.

\section{Conclusion}\label{sec7}

We have considered the classical limit of two independent quantum harmonic oscillators, with
$\hbar$ as quantization parameter, whose position coordinates are themselves non-commuting operators,
with non-commutative deformation parameter $\theta$.
This non-commutative quantum model allows for the construction of creation and annihilation operators with a corresponding
Weyl algebra; we have thus studied the classical limit by means of the so-called anti-Wick quantization scheme that uses
coherent states to map a commutative $C^*$ algebra of continuous functions into the non-commutative $C^*$
algebra generated by the Weyl operators and to map these operators back to continuous functions.

Three possibilities appear to implement the scheme:
\begin{enumerate}
\item
to link $\hbar$ and $\theta$ so that one may consider the
classical limit $\mu_{\hbar,\theta}\to0$;
\item
let $\hbar\to0$ first so to get to a non-quantum non-commutative
system and then let $\theta\to0$;
\item
to let $\theta\to0$ first so to get to standard quantum mechanics
and then let $\hbar\to0$.
\end{enumerate}

In the given model, the first possibility corresponds to 
an anti-Wick quantization procedure which quantizes a $C^*$ algebra
of continuous functions over $\RI^4$ and de-quantizes it back to the 
same algebra.
In the second case, when $\theta\to0$, one gets the Weyl algebra 
of two standard quantum oscillators and then the continuous functions 
over $\RI^4$ when $\hbar\to0$.
Instead, the third possibility is such that $\hbar\to 0$ first yields 
a quantization scheme of a $C^*$ algebra of continuous functions over 
$\RI^2$ (not $\RI^4$) and then $\theta\to0$ maps the Weyl algebra generated 
by the non-commuting position coordinates of the two oscillators back to 
the continuous functions over $\RI^2$.
The non-exchangeability of the two limits
$$
\lim_{\theta\to0}\lim_{\hbar\to0}\,\neq\,\lim_{\hbar\to0}\lim_{\theta\to0}\ ,
$$
becomes even more  evident when one considers the dynamics of 
the non-commutative quantum oscillators generated by a quadratic 
Hamiltonian in the
non-commutative quantum creations and annihilation operators. 
In this case, while the classical limit performed according to 
the first and second possibilities yields the classical 
Hamiltonian dynamics of two identical, independent harmonic oscillators, 
in the third case the non-commutative non-quantum dynamics does 
not survive the classical limit, but for the fact that it further 
reduces to $\RI$ the space of definition of continuous functions 
initially defined on $\RI^4$.

{\bf ACKNOWLEDGMENTS}

L. G. gratefully acknowledges the support 
of the Abdus Salam International Centre for 
Theoretical Physics (ICTP), Trieste.

\section{Appendix: details of the function $F_{\hbar,\theta}$ 
in section \ref{sec5}}

Let us consider
\begin{equation}
\label{App1}
F_{\hbar,\theta}(r) = \frac{J}{\pi^2}
\int_{\mathbb{R}^4} {\rm d}r'\, F(r')\,\left\vert 
\langle z(r)\vert z(r')\rangle\right\vert^2
=\frac{\hbar^2}{4\mu_{\hbar,\theta}^4\pi^2}
 \int_{\mathbb{R}^4} {\rm d}r' \, F(r')\, e^{-E(r,r')}\ ,
 \end{equation}
where
\begin{eqnarray}
\nonumber
E(r,r')&=&|z_1(r)-z_1(r')|^2 +|z_2(r)-z_2(r')|^2\\
\nonumber
&=&\frac{1}{4\mu^2_{\hbar,\theta}(\lambda_+ +\lambda_{-})^2}\times \Big\{
\Big(\lambda_{-}^2K_+ +\lambda_+^2K_{-}\Big)(x_1-x'_1)^2
+\hbar^2\Big(K_+ + K_{-}\Big)(y_2 -y'_2)^2\\
\nonumber
&+& \Big(\lambda_{-}^2K_+ +\lambda_+^2K_{-}\Big)(x_2-x'_2)^2
+\hbar^2\Big(K_+ + K_{-}\Big)(y_1 -y'_1)^2\\
&-& 2\hbar\Big(\lambda_{-}K_{+} - \lambda_{+}K_{-}\Big)(x_1-x_1')(y_2-y_2')
+2\hbar\Big(\lambda_{-}K_{+} - \lambda_{+}K_{-}\Big)(x_2-x_2')(y_1-y_1')
\Big\}\ .
\label{expon}
\end{eqnarray}
First, by setting
$u_i = (x_i -x_i')$ and $v_i = (y_i -y_i')$, $i = 1,2$, one gets
\begin{eqnarray}
\label{hteta}
F_{\hbar,\theta}(r) &=& 
\frac{\hbar^2}{4\mu_{\hbar,\theta}^4\pi^2}
\int_{\mathbb{R}^4} {\rm d}u_1{\rm d}
u_2{\rm d}v_1{\rm d}v_2\,
F(x_1 + u_1, x_2+u_2,y_1+v_1,y_2+v_2)
\,e^{-D(u_1,u_2,v_1,v_2)}\\\nonumber
\label{hteta1}
D(u_1,u_2,v_1,v_2) &=& \frac{1}{4m
\omega\mu_{\hbar,\theta}^2
\sqrt{4\hbar^2 +m^2\omega^2\theta^2}}
\Big\{4m^2\omega^2\hbar^2 u_1^2 + 
(4\hbar^2 + 2m^2\omega^2\theta^2)v_2^2
-4m^2\omega^2\hbar\theta u_1v_2\\
&+&
4m^2\omega^2\hbar^2 u_2^2 + 
(4\hbar^2 + 2m^2\omega^2\theta^2)v_1^2+4m^2
\omega^2\hbar\theta u_2v_1\Big\}\ .
\end{eqnarray}
Next, the change of variables
$\displaystyle
\bar{u}_i = \frac{4m\omega\hbar}{\mu_{\hbar,\theta}}u_i$, $\displaystyle
\bar{v}_i = \frac{2\sqrt{4\hbar^2 + 2m^2\omega^2\theta^2}}{\mu_{\hbar,\theta}}v_i$, $i = 1,2$
yields
\begin{eqnarray}
\nonumber
&&
F_{\hbar,\theta}(r)=\frac{1}{16m^2
\omega^2(4\hbar^2 +2m^2\omega^2\theta^2)}
\int_{\mathbb{R}^4} {\rm d}\bar{u}_1 {\rm d}
\bar{u}_2 {\rm d} \bar{v}_1 {\rm d}\bar{v}_2\,
e^{-G(\bar{u}_1,\bar{u}_2,\bar{v}_1,\bar{v}_2)}\ 
\times\\\nonumber
\\
\nonumber
\\
&&\hskip .5cm
\times
F\Big(x_1+\frac{\mu_{\hbar,\theta}}{2m\omega\hbar}\bar{u}_1, x_2+
\frac{\mu_{\hbar,\theta}}{2m\omega\hbar}\bar{u}_2,
y_1 +\frac{\mu_{\hbar,\theta}}{\sqrt{4\hbar^2 + 2m^2\omega^2\theta^2}}\bar{v}_1,
y_2 +\frac{\mu_{\hbar,\theta}}{\sqrt{4\hbar^2 + 2m^2\omega^2\theta^2}}\bar{v}_2\Big)
\label{exp}
\\
\nonumber
\\
&&
G(\bar{u}_1,\bar{u}_2,\bar{v}_1,\bar{v}_2)=
\frac{\gamma_+(\bar{u}_1 -\bar{v}_2)^2
+\gamma_{-}(\bar{u}_1 +\bar{v}_2)^2
+\gamma_+(\bar{u}_2 +\bar{v}_1)^2
+\gamma_{-}(\bar{u}_2 -\bar{v}_1)^2}{4m\omega\sqrt{4\hbar^2 +2m^2\omega^2\theta^2}}\ ,
\end{eqnarray}
with $\displaystyle \gamma_\pm = \frac{1}{2}\Big(1 \pm\frac{m\omega\theta}
{\sqrt{4\hbar^2 + 2m^2\omega^2\theta^2}}\Big)$.

The expression in equation (\ref{exp}) can be diagonalized by setting
\begin{eqnarray*}
 w_1 &=& \sqrt{\frac{4m\omega}{\gamma_+}}\sqrt[4]{4\hbar^2 + 2m^2\omega^2\theta^2}
 (\bar{u}_1 -\bar{v}_2)\ ,\quad
 w_2 = \sqrt{\frac{4m\omega}{\gamma_{-}}}\sqrt[4]{4\hbar^2 + 2m^2\omega^2\theta^2}
 (\bar{u}_1 + \bar{v}_2)\\
 w_3 &=& \sqrt{\frac{4m\omega}{\gamma_+}}\sqrt[4]{4\hbar^2 + 2m^2\omega^2\theta^2}
 (\bar{u}_2 +\bar{v}_1)\ ,\quad
 w_4 = \sqrt{\frac{4m\omega}{\gamma_{-}}}\sqrt[4]{4\hbar^2 + 2m^2\omega^2\theta^2}
 (\bar{u}_2 -\bar{v}_1)
\end{eqnarray*}
This finally yields
\begin{eqnarray}
&&
F_{\hbar,\theta}(r)=\frac{1}{\pi^2}
\int_{\mathbb{R}^4}{\rm d}w_1{\rm d}w_2{\rm d}w_3{\rm d}w_4\,
e^{-(w_1^2 +w_2^2 +w_3^2 +w_4^2)}\,\times\\
\nonumber
&&\hskip.5cm
\times\, F\Bigg(x_1+\frac{\mu_{\hbar,\theta}\sqrt[4]{4\hbar^2 +
 m^2\omega^2\theta^2}}{2\sqrt{m\omega}\hbar}
 \Big(\frac{w_1}{\sqrt{\gamma_+}}+\frac{w_2}{\sqrt{\gamma_{-}}}\Big)\ ,
 x_2+\frac{\mu_{\hbar,\theta}\sqrt[4]{4\hbar^2 + m^2\omega^2
 \theta^2}}{2\sqrt{m\omega}\hbar}
 \Big(\frac{w_3}{\sqrt{\gamma_+}}+\frac{w_4}{\sqrt{\gamma_{-}}}\Big)\ ,\\
 \nonumber
&&\hskip .5cm
y_1+\frac{\mu_{\hbar,\theta}\sqrt{m\omega}\sqrt[4]{4\hbar^2 + m^2\omega^2
 \theta^2}}{\sqrt{4\hbar^2 + 2m^2\omega^2\theta^2}}
 \Big(\frac{w_3}{\sqrt{\gamma_+}}-\frac{w_4}{\sqrt{\gamma_{-}}}\Big),
 y_2+\frac{\mu_{\hbar,\theta}\sqrt{m\omega}\sqrt[4]{4\hbar^2 + m^2\omega^2
 \theta^2}}{\sqrt{4\hbar^2 + 2m^2\omega^2\theta^2}}
 \Big(\frac{w_2}{\sqrt{\gamma_-}}-\frac{w_1}{\sqrt{\gamma_{+}}}\Big)
 \Bigg)\ .
\end{eqnarray}

Then, one obtains  equation~(\ref{wickl}) by means of the
following change of variables in equation (\ref{NCQ2Qa}):
\begin{eqnarray}
 u_1 = \frac{w_1+w_2}{\sqrt{2}},\quad
 u_2 = \frac{w_3+w_4}{\sqrt{2}},\quad
 v_1 = \frac{w_3 - w_4}{\sqrt{2}},\quad
 v_2 = \frac{w_2 - w_1}{\sqrt{2}}\ .
\end{eqnarray}
The Jacobian for this change of variable
is $J =1$ and $w_1^2 +w_2^2 + w_1^2 +w_2^2 
= u_1^2 +u_2^2 +v_1^2 + v_2^2$.

\end{document}